%% file: main.tex
\documentclass[final,12pt]{llncs}
\usepackage[utf8]{inputenc}

\input{macros.sty}
\begin{document}

\title{Drawing non-planar graphs with rotation systems on the Klein bottle}

\author{
François Doré
\and
Enrico Formenti
}

\institute{Université Côte d'Azur, CNRS, I3S, France}

\maketitle

\begin{abstract}

This paper provides a linear time algorithm in the number of edges that,
given a simple 3-connected non-planar graph $G$ with a Klein bottle rotation system, outputs  a straight line drawing of $G$ with no crossings on the flat Klein bottle.
\end{abstract}

\keywords{Straight-line drawing \and Rotation Systems \and Graph Embedding \and Non-orientable Surfaces}

\input{content.tex}

\bibliographystyle{plain}
\bibliography{main}
\end{document}

%% file: content.tex
\section{Introduction}

Wagner~\cite{WAGNER_BEFORE_FARY} and F\'ary~\cite{FARY_STRAIGHTLINE_PLANAR},
independently, proved that simple planar graphs admit a straight-line representation on the plane. Mohar extended this result to flat surfaces (2-polytopes where distinct sides are identified as one) \ie the cylinder, the Möbius band, the flat torus and the flat Klein bottle~\cite{MOHAR_FARY_THEOREM_ON_OTHER_SURFACES}. However, even knowing that they exists, finding these representations is not an easy task. Read proposed an algorithm to create a straight-line representation of planar graphs given their rotation systems~\cite{READ_STRAIGHTLINE_PLANAR}. This algorithm has then been adapted to the torus by Kocay et al.~\cite{KOCAY_STRAIGHTLINE_TORUS} but, to the best of our knowledge, similar algorithms for non-orientable surfaces, especially for the Klein bottle, have not been much studied.

\begin{figure}
\centering
\begin{tikzpicture}[scale=1.5]
\begin{scope}[xshift=-1 cm]
    \draw[doublearrow] (0,0)--(1,0);
    \draw[middlearrow] (1,0)--(1,1);
    \draw[doublearrow] (0,1)--(1,1);
    \draw[middlearrow] (0,0)--(0,1);
\end{scope}
\begin{scope}[xshift=1cm]
    \draw[doublearrow] (0,0)--(1,0);
    \draw[middlearrow] (1,0)--(1,1);
    \draw[doublearrow] (0,1)--(1,1);
    \draw[middlearrow] (0,1)--(0,0);
\end{scope}
\end{tikzpicture}
\caption{Flat representation of the torus (left) and the Klein bottle (right). Pairs of identified sides of the $[0,1]$ Square have here the same number of markings on them but we will assume that these pairs always concern two opposite sides} \label{fig:torus-and-klein}
\end{figure}
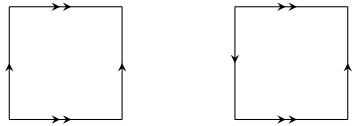

Although the torus and the Klein bottle seem similar at a first glance, there are classes of graphs that can be embedded in the torus but not in the Klein bottle and conversely. Riskin studied in detail the nonembeddability of toroidal graphs on the Klein bottle and stated, in his conclusions, the following conjecture ~\cite{RISKIN_NONEMBEDDABILITY_TORUS_KLEIN}:
\begin{quote}    
Klein bottle polyhedral maps with four disjoint homotopic noncontractible circuits of the type that collapse the Klein bottle to the pinched torus are not toroidal.
\end{quote} 
The class of graphs described by Riskin's conjecture is not the only one which has a Klein bottle embedding but are not toroidal. Indeed, Figure~\ref{fig:klein-but-not-torus} shows a Klein-embeddable graph which does not satisfy the hypothesis of Riskin's conjecture but have a minor homeomorphic to a torus obstruction listed in the database made by Myrvold and Woodcock~\cite{MYRVOLD_TORUS_OBSTRUCTIONS}.
\begin{figure}
\centering
\begin{tikzpicture}
\begin{scope}[rotate=-90,scale=.5]
    \draw[middlearrow] (4,0)--(0,0);
    \draw[middlearrow] (4,0)--(4,16);
    \draw[middlearrow] (4,16)--(0,16);
    \draw[middlearrow] (0,16)--(0,0);

    \foreach \x in {0,...,7}{
        \pgfmathtruncatemacro{\lbl}{\x+8}
        \node[gnode,minimum size=10](\x) at (1,1+2*\x){\tiny \x};
        \node[gnode,minimum size=10](\lbl) at (3,1+2*\x){\tiny \lbl};
    }

    \foreach \u/\v in {0/8,1/9,5/6,2/10,13/14,7/15} {
        \draw[gedge,densely dashed](\u)--(\v);
    }
    \foreach \u/\x/\y in {0/0/1,3/0/7,4/0/9,11/4/7,12/4/9,15/4/15} {
        \draw[gedge,densely dashed](\u)--(\x,\y);
    }
    \draw[gedge](1,0)--(0)--(1)--(2)--(3)--(4)--(5)--(13)--(12)--(11)--(10)--(9)--(8)--(3,0);
    \draw[gedge](1,16)--(7)--(6)--(14)--(15)--(3,16);
    \draw[gedge](1,16)--(7)--(6)--(14)--(15)--(3,16);
    \draw[gedge](4)--(12);
    \draw[gedge](3)--(11);
    
    \foreach \u/\x/\y in {1/0/3,2/0/5,5/0/11,6/0/13,7/0/15,8/4/1,9/4/3,10/4/5,13/4/11,14/4/13} {
        \draw[gedge](\u)--(\x,\y);
    }
\end{scope}
\begin{scope}[xshift=9.5cm,yshift=0.75cm,scale=.5]
    \foreach \x/\y/\l in {0/0/14,0/1/15,0/2/8,0/3/9,0/4/10,0/5/11,0/6/12,0/7/13,1/0/6,1/1/7,1/2/0,1/3/1,1/4/2,1/5/3,1/6/4,1/7/5}{
        \node[gnode,minimum size=10](\l) at (2*\x,-\y){\tiny \l};
    }
    \draw[gedge](14)--(15)--(8)--(9)--(10)--(11)--(12)--(13)--(5)--(4)--(3)--(2)--(1)--(0)--(7)--(6)--cycle;
    \draw[gedge](12)--(4);
    \draw[gedge](11)--(3);
    \draw[gedge](8) --(7);

    \draw[gedge] (14) to [out=45 ,in=45 ,looseness=2] (1);
    \draw[gedge] (6)  to [out=135,in=135,looseness=2] (9);
    \draw[gedge] (5)  to [out=225,in=225,looseness=2] (10);
    \draw[gedge] (13) to [out=315,in=315,looseness=2] (2);
    
\end{scope}

\end{tikzpicture}
\caption{An Klein embedding of a square grid which is non-toroidal. The subgraph homeomorphic to a torus obstruction is drawn by solid lines (left) and its embedding the plane (right).} \label{fig:klein-but-not-torus}
\end{figure}
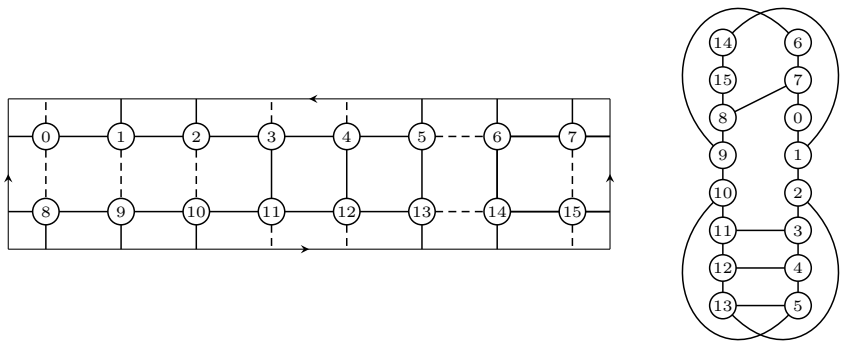
Indeed, the whole class of square grids of size $m\times n$ on the Klein bottle where $m$ runs along the non-inverted side and $n$ runs along the inverted side
will be non-toroidal when $m\geq2$ and $n\geq8$ since they contain the sub-grid of Figure~\ref{fig:klein-but-not-torus}.

These facts motivate the search for algorithms which given a
rotation system can build a straight-line drawing on the Klein bottle.
This is the main matter of the next sections. We stress that some of
the algorithms are sufficiently generic to be extendable to other
surfaces.

\section{Rotation systems and the Klein bottle}

We will use the standard definitions and concepts from graph topology (for more details see, for instance, the book of Mohar and  Tho\-mas\-sen~\cite{MOHAR_GRAPH_ON_SURFACES}). An \emph{embedding} $\emb$ of a graph $G=\structure{V,E}$ on a surface $S$ is a representation of $G$ in which vertices are points on $S$ and edges are simple curves homeomorphic to $[0,1]$ over $S$. This representation is such that endpoints of a curve associated with an edge must coincide with the endpoints of the edge, no curve representing an edge contains more than two vertices and no two curves intersect at a common interior point. A \emph{face} of an embedding is a maximal contiguous region of $S-\emb$. An embedding is called \emph{cellular} if all its faces are homeorphic to disks. The \emph{Euler's characteristic} $\chi(\emb)$ of an embedding is equal to $|V|-|E|+|F|$, where $F$ is the set of all the faces of $\emb$.

It is well known that \emph{orientable rotation systems} of a graph $G$, or \emph{pure} rotation system, induce a unique embedding of $G$ on an orientable surface $S$ up to embedding equivalence, but this relation, is not necessarily the same in the non-orientable case. Below, we define \emph{general} (or \emph{non-orientable}) rotation systems. 

\begin{definition}
Let $G=\structure{V,E}$ be a graph, its \emph{general rotation system} $\rot$ is the structure $\structure{\pi,w}$, where $\pi=\set{\pi_v\mid v \in V}$ is the ordered adjacency of each vertex, and $w=\set{w_e \mid e\in E}$ is the sign of each edge (indicating if they are twisted or not).
\end{definition}

From now on, we will consider that nodes labelled by integers in the
set \set{1,\ldots,|V|}. Thus, a node $u$ is lower than a node $v$ if the same relation holds for their labels.

The rotation system $\Pi_2(G)$ is the \emph{flip} of $\Pi_1(G)$
iff the adjacencies of all of its vertices are reversed \wrt those
of $\Pi_1(G)$.
Two orientable rotation systems $\Pi_1(G)$ and $\Pi_2(G)$ are \emph{equivalent} if letting $\Pi$ be $\Pi_1(G)$ or its flip,
for each node of $\Pi$, there exists a cyclic permutation which transforms the ordered adjacency into the one of $\Pi_2(G)$. In order to test the equivalence of two
orientable rotation systems $\Pi_1(G)$ and $\Pi_2(G)$, one can 
consider the minimal cyclic permutation of the ordered adjacency of
$\Pi_1(G)$ and $\Pi_2(G)$ and then check the adjacency of the lowest node, if the second node is greater than the last, we take the \emph{flip}.
With this procedure, we can easily check if two rotation systems of one graph are equivalent or not. However, with non-orentable rotation systems, this is no longer sufficient.

\newcommand{\flatKleinBottle}{
    \draw[beforearrow] (0,0)--(4,0);
    \draw[beforearrow] (4,0)--(4,4);
    \draw[beforearrow] (0,4)--(4,4);
    \draw[beforearrow] (0,4)--(0,0);
}

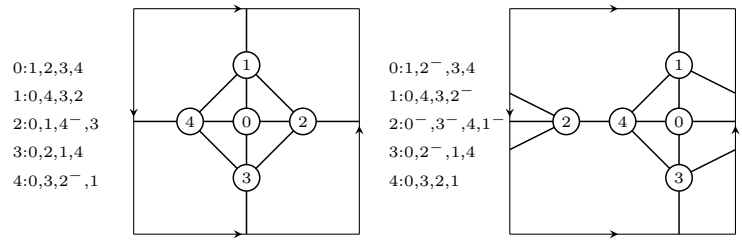
\begin{figure}
\centering
\begin{tikzpicture}
\begin{scope}[xshift=-2.5cm,scale=.75]
    \flatKleinBottle
    \node[gnode,minimum size=10](0) at (2,2) {\tiny 0};
    \node[gnode,minimum size=10](1) at (2,3) {\tiny 1};
    \node[gnode,minimum size=10](2) at (3,2) {\tiny 2};
    \node[gnode,minimum size=10](3) at (2,1) {\tiny 3};
    \node[gnode,minimum size=10](4) at (1,2) {\tiny 4};

    \draw[gedge](2,4)--(1)--(2)--(3)--(4)--(1)--(0)--(3)--(2,0);
    \draw[gedge](0,2)--(4)--(0)--(2)--(4,2);

    \node[anchor = south west]() at (-2.3,-.35+3)  {\tiny $0$:$1$,$2$,$3$,$4$};
    \node[anchor = south west]() at (-2.3,-.35+2.5){\tiny $1$:$0$,$4$,$3$,$2$};
    \node[anchor = south west]() at (-2.3,-.35+2)  {\tiny $2$:$0$,$1$,$4^-$,$3$};
    \node[anchor = south west]() at (-2.3,-.35+1.5){\tiny $3$:$0$,$2$,$1$,$4$};
    \node[anchor = south west]() at (-2.3,-.35+1)  {\tiny $4$:$0$,$3$,$2^-$,$1$};
\end{scope}
\begin{scope}[xshift=2.5cm,scale=.75]
    \flatKleinBottle
    \node[gnode,minimum size=10](0) at (3,2) {\tiny 0};
    \node[gnode,minimum size=10](1) at (3,3) {\tiny 1};
    \node[gnode,minimum size=10](2) at (1,2) {\tiny 2};
    \node[gnode,minimum size=10](3) at (3,1) {\tiny 3};
    \node[gnode,minimum size=10](4) at (2,2) {\tiny 4};

    \draw[gedge](3,4)--(1)--(0)--(3)--(3,0);
    \draw[gedge](0,2)--(2)--(4)--(0)--(4,2);
    \draw[gedge](4,2.5)--(1)--(4)--(3)--(4,1.5);
    \draw[gedge](0,2.5)--(2)--(0,1.5);

    \node[anchor = south west]() at (-2.3,-.35+3)  {\tiny $0$:$1$,$2^-$,$3$,$4$};
    \node[anchor = south west]() at (-2.3,-.35+2.5){\tiny $1$:$0$,$4$,$3$,$2^-$};
    \node[anchor = south west]() at (-2.3,-.35+2)  {\tiny $2$:$0^-$,$3^-$,$4$,$1^-$};
    \node[anchor = south west]() at (-2.3,-.35+1.5){\tiny $3$:$0$,$2^-$,$1$,$4$};
    \node[anchor = south west]() at (-2.3,-.35+1)  {\tiny $4$:$0$,$3$,$2$,$1$};
\end{scope}
\end{tikzpicture}
\caption{Two equivalent embeddings with different rotation systems. A minus symbol at the exponent means that the edge is going through the inverted side of the Klein bottle.} \label{fig:iso-embs}
\end{figure}

As shown in Figure~\ref{fig:iso-embs}, moving one node $v$ (the node $2$ in this case) through an inverted side, implies changes on the rotation system (two changes in our figure). First, the adjacency of $v$ is \emph{flipped}, meaning the order of its adjacency is reversed. Second, all the edges incident to $v$ have their signs changed. We call this operation \emph{switching} the vertex $v$.

Two rotations systems $\Pi_1(G)$ and $\Pi_2(G)$, such that $\Pi_1(G)$ can be obtained from $\Pi_2(G)$ by a sequence of flip/switch are said to be \emph{switch-equivalent}. We will denote this fact by  $\Pi_1(G)\iso \Pi_2(G)$.
This means that we can freely and independently switch any vertex without altering the rotation system. To have a rotation system independent of the switches, we can switch the nodes $v$ whose $\pi_v$ needs to be flipped, handle the changes for the edge signs accordingly and then consider the minimal cyclic permutations as for the orientable case. An outline of this formatting method is shown in Algorithm~\ref{alg:canonical}.

\begin{algorithm}[ht!]
\caption{Format rotation system}\label{alg:canonical}
\begin{algorithmic}[1]
\Require $G = \structure{V,E}$,$\rot$
\Ensure $\rot$
\For{$v \in V$}
    \State $u \gets \min(\pi_v)$
    \If {$\pi_v.successor(u)>\pi_v.predecessor(u)$}
        \For{$e \in E(v)$}
            \State $w_e \gets -w_e$
        \EndFor
        \State $\pi_v \gets flipped(\pi_v)$
    \EndIf
    \State $\pi_v \gets rotated(\pi_v)$
\EndFor
\end{algorithmic}
\end{algorithm}

The functions $\pi_v.successor(u)$ (resp. $\pi_v.predecessor(u)$) return the node before (resp. after) $u$ in the ordered adjacency of $v$. The functions $flipped(\pi_v)$ and $rotated(\pi_v)$ output respectively the reverse and the minimal cycle permutation of $\pi_v$.

The rotation system formatting algorithm runs through all the vertices and changes the sign of each edge at most twice, therefore its complexity is $O(|V|+|E|)$.

This algorithm will be particularly useful when enumerating the embeddings needed as a base for our drawing algorithm.

\section{Enumerating the embeddings}

In this section we are going to provide an enumeration algorithm
for the embeddings of a graph into a generic surface. We stress that
enumeration is not an easy task in general. Indeed, given a graph
$G$ with $n$ nodes and $m$ edges, an upper bound for the number of labelled embeddings of $G$ is given by \[2^m\prod_{0<i\leq n} (d_i-1)!\]
where $d_i$ is the degree of node $i$. Therefore, since we assumed
$d_i>2$ the above bound is certainly larger than $2^{m+n}$ and hence for large values of $m$ or $n$, labelled enumeration can be practically unfeasible. We therefore prefer to enumerate 
the unlabelled embeddings up to isomorphism and switch-equivalence. Some theoretical results exist in this domain, but no tight bounds on the number of these embeddings nor a way to generate them efficiently, to our knowledge, is known (see~\cite{CHAPUY_COUNTING_UNICELLULAR_MAPS,GROSS_ENUMERATING_GRAPH_EMBEDDINGS} for further details).

To enumerate all possible unlabelled embeddings of a graph $G$ on the Klein bottle, one can enumerate all the rotation systems of $K$, changing both adjacencies and edge twists. For each possible rotation system $\rot$, compute its genus (or Euler characteristic) with a face-walking algorithm to keep only the ones with $\chi(\rot)=0$, and store the minimal (lexicographically ordered) formatted form among all the relabellisations of $G$.

Although for small graphs, the complete enumeration can be done quite easily, few optimizations can be done to minimize the number of times a canonical form is generated. 

Firstly, it is not mandatory to test all the possible relabellisations of $\rot$ when computing its canonical form. One can simply test permutations conserving the automorphism groups\footnote{See the Nauty library of McKay for automorphism and isomorphism algorithms that have proved their worth and that are usable in practice~\cite{MCKAY_GRAPH_ISOMORPHISM}.}. 

A second step in this direction consists in exploiting a notion coming from the domain of signed graphs, namely the \emph{frustration} of a graph.


\begin{definition}
Let $G=\structure{V,E}$ be a graph and $\rot$ an embedding of $G$. The \emph{frustration} $f(\rot)$ of $\rot$ is the minimum number of twisted edges after any sequence of switches.
\end{definition}

\begin{figure}
\centering
\begin{tikzpicture}[scale=.8]
  \node[](k5_0) at (-4,0) {%
    \begin{tikzpicture}[scale=.8]
        \foreach \i in {0,1,...,4} {
            \node[gnode] (\i) at ({cos(90-360/5*\i)},{sin(90-360/5*\i)}){};
        }
        \draw[gedge,densely dotted](0)--(1)--(2)--(3)--(4)--(0)--(2)--(4)--(1)--(3)--(0);
    \end{tikzpicture}
  };
  \node[](k5_1) at (0,0) {%
    \begin{tikzpicture}[scale=.8]
        \foreach \i in {0,1,...,4} {
            \node[gnode] (\i) at ({cos(90-360/5*\i)},{sin(90-360/5*\i)}){};
        }
        \draw[gedge,densely dotted](1)--(2)--(3)--(4)--(1)--(3);
        \draw[gedge,densely dotted](4)--(2);
        \draw[gedge](1)--(0)--(3);
        \draw[gedge](2)--(0)--(4);
    \end{tikzpicture}
  };
  \node[](k5_2) at (4,0) {%
    \begin{tikzpicture}[scale=.8]
        \foreach \i in {0,1,...,4} {
            \node[gnode] (\i) at ({cos(90-360/5*\i)},{sin(90-360/5*\i)}){};
        }
        \draw[gedge,densely dotted](0)--(1)--(4)--(0);
        \draw[gedge,densely dotted](2)--(3);
        \draw[gedge](1)--(2)--(4)--(3)--(1);
        \draw[gedge](2)--(0)--(3);
    \end{tikzpicture}
  };

  \draw[-stealth] (k5_0) to [bend left = 35] (k5_1);
  \draw[-stealth] (k5_1) to [bend left = 35] (k5_2);
  \draw[-stealth] (k5_2) to [bend left = 35] (k5_1);
  \draw[-stealth] (k5_1) to [bend left = 35] (k5_0);
  \draw[-stealth] (k5_2) to [out=35,in=325,looseness=3.5] (k5_2);

\end{tikzpicture}
\caption{Three switch-equivalent signatures of $K_5$ (up to isomorphism) having a frustration of 4. 
Solid (resp., dotted) lines represent positive (resp., negative or twisted) edges. Curved arrows denote the possible switch-transitions between signatures.}
\label{fig:definition-frustration}
\end{figure}
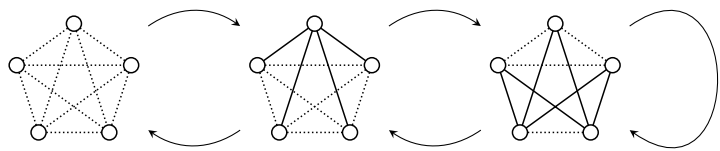

Now, we are going to define the frustration for standard (\ie unsigned) graphs.

\begin{definition}
Let $G=\structure{V,E}$ be a graph. The \emph{frustration} $f(G)$ of $G$ is the maximum of the frustrations $f(\rot)$ over all possible $\rot$ of $G$.
\end{definition}

Having the notion of frustration in mind, there is no need to check for all the configurations for edge signs. Indeed, it is well known that the frustration of any graph is at most half the number of edges~\cite{ALON_SIGNED_GRAPH_BOOK}. So only a subset of all possible signatures is needed when enumerating the labelled embeddings.


\begin{theorem} \label{th:non-cellular}
Let $G=\structure{V,E}$ be a graph with $\rot$ its embedding (potentially with twisted edges) and with $\chi(\rot)=0$. If $f(\rot)=0$, then the corresponding embedding on the Klein bottle is not cellular.
\end{theorem}
\begin{proof}
Let $G$ be a graph satisfying the hypothesis. Assume $f(\rot)=0$. 
Then, there exists a way to move the nodes of $G$ on the Klein bottle in such a way that there are no twisted edges anymore, \ie we can produce an embedding $\Pi'(G)$ in which all the edges are either fully contained in $[0,1]^2$ or passing through the non-inverted identified sides. Thus, there is at least one face, the one containing the inverted side, which is not homeomorphic to a disk.
\end{proof}

Finally, the following theorem provides a criterion to exclude some rotation systems which are, in practice, not embeddable in the Klein bottle.

\begin{theorem} \label{th:invalid-rotation-system}
Let $G=\structure{V,E}$ be a non-planar graph with $\rot$. If $\chi(\rot)=0$ and $f(\rot)=0$, then there is no embedding of $G$ corresponding to $\rot$ in the Klein bottle.
\end{theorem}

\begin{proof}
By Theorem~\ref{th:non-cellular}, if $f(\rot)=0$, there is a face homeomorphic to a cylinder. By cutting through it, the supposed embedding of $G$ becomes a cylindrical one. However, each cylindrical embedding is necessarily planar (by identifying one of the two borders as one point) which is a contradiction. 
\end{proof}

\begin{algorithm}[ht!]
\caption{Enumerating distinct Klein bottle embeddings}\label{alg:enumeration}
\begin{algorithmic}[1]
\Require $G = \structure{V,E}$
\State $X_{\text{False}} \gets set()$
\State $X_{\text{All}} \gets set()$
\For{$edgeMask=0$ \textbf{to} $2^{|E|}-1$} \label{alg:enum:sign-loop}
    \For{$(\sigma_0,\ldots,\sigma_{|V|-1}) \in \prod_{0<i\leq|V|} S_{d_i-1} $} \label{alg:enum:perm-loop}
        \State $\rot.clear()$ \label{alg:enum:construction-start}
        \State $sortedEdges \gets sorted(E)$
        \For{$i=0$ \textbf{to} $|E|-1$}
            \State $e \gets sortedEgdes[i]$
            \State $w_e \gets 1|-(edgeMask>>i\&1)$
        \EndFor
        \For{$v \in V$}
            \State $u_{min} \gets min(V(v))$
            \State $\pi_v \gets [u_{min}]+sorted(V(v)\backslash \set{u_{min}})^{\sigma_v}$ \label{alg:enum:perm-application}
        \EndFor \label{alg:enum:construction-end}
        \If{$\chi(\rot) = 0$}
            \State $\Pi \gets min(\set{(\rot^{\sigma}).format() \mid \sigma \in S_{|V|}})$ \label{alg:enum:canonical}
            \If{$edgeMask = 0$}
                \State $X_{\text{False}}.add(\Pi)$
            \EndIf
            \State $X_{\text{All}}.add(\Pi)$
        \EndIf
    \EndFor
\EndFor
\State $\textbf{return } X_{\text{All}} \backslash X_{\text{False}}$
\end{algorithmic}
\label{algo:enum-distinct-embeddings}
\end{algorithm}

At this point, all the material for the enumeration
algorithm are ready. Let us comment how
Algorithm~\ref{alg:enumeration} works. First of all, remark that $S_n$ (line~\ref{alg:enum:perm-loop}) denotes the set of permutations of $n$ elements and $X^{\sigma}$ (lines~\ref{alg:enum:perm-application} and~\ref{alg:enum:canonical}) denotes the application of the permutation $\sigma$ on the set $X$. Lines~\ref{alg:enum:sign-loop}-\ref{alg:enum:perm-loop} build the combinatorial objects needed for the enumeration and lines~\ref{alg:enum:construction-start} through~\ref{alg:enum:construction-end} construct a rotation system given those combinatorial objects. Line~\ref{alg:enum:canonical} finds the canonical form of the constructed rotation system. The set $X_{\text{False}}$ stores the pseudo-valid (\ie with Euler's characteristic equal $0$) systems but with a frustration of $0$, \ie the ones whose a corresponding embedding in the Klein bottle does not exists (see Theorem~\ref{th:invalid-rotation-system}), and $X_{\text{All}}$ store the pseudo-valid ones regardless of their frustration. For sake of clarity, the two optimisations detailed above are not integrated in the algorithm but they could intervene respectively lines~\ref{alg:enum:sign-loop} and~\ref{alg:enum:canonical}.

\section{The straight-line drawing algorithm}

Our drawing algorithm uses the idea of \emph{embedding extension}:  extract a subgraph which we know how to draw and then extend the drawing by adding the missing nodes and edges (see~\cite{MYRVOLD_TORUS_OBSTRUCTIONS,JUVAN_MOHAR_ALGORITHM_TORUS_N3,DEMOUCRON_MALGRANGE_PERTUISET_DMP} for instance).
Since the main motivation of the algorithm is to draw non-planar and even non-toroidal graphs by a straight-line drawing, we can hence take a Kuratowski subgraph, \ie $K_5$ or $K_{3,3}$, as a base.
However, to be sure to be able to draw this subgraph according to the given rotation system, we have first to determine all the ways to draw these two graphs on the Klein bottle. We will also make the standard assumption that our input graph is 3-connected. If not, dummy edges can be added to make $G$ become 3-connected, then removed at the very end.

We introduce a notation for the drawing of a graph $G$, usable both in theory and implementation-wise.

\begin{definition}
Let $G=\structure{V,E}$ be a graph, we define its drawing $\Gamma(G)=\structure{\gamma,\delta}$, with $\gamma=\set{\gamma_v \in [0,1]^2 \mid v \in V}$, the coordinate of the node $v$ in the Klein bottle, and $\delta=\set{\delta_e \in \Z^2 \mid e\in E}$, the shifts of each edge defining if they go through one or multiple identified sides of the Klein bottle.
\end{definition}

One coordinate, $x$ or $y$ is chosen to correspond to the inverted side, without loss of generality, say $x$. Let consider a node $v$ with coordinates $(\gamma_{x_v},\gamma_{y_v})$ having a neighbour $u$ linked by an edge with a $(\delta_x,\delta_y)$ shift. To compute the coordinates of $u$ relative to $v$, we will consider two cases:

\begin{equation*}
\gamma_{v \rightarrow u}=
\begin{cases}
(\delta_x+\gamma_{x_u},\delta_y+\gamma_{y_u}) & \text{$\delta_x$ is even} \\
(\delta_x+\gamma_{x_u},\delta_y+1-\gamma_{y_u}) & \text{$\delta_x$ is odd}
\end{cases}
\end{equation*}

These relative coordinates are mainly used to draw the edges between the nodes or to compute the center of a face. This can intuitively be extended to non-adjacent nodes, in which case we sum up the shifts of the edges on the path between the nodes. We assume in the following that the paths in question are clear form the context.
\smallskip

As explained before, the base of our drawings are the Kuratowski subgraphs present in our input graph. As a preprocessing step, all the possible embeddings of $K_5$ and $K_{3,3}$
are precomputed in advance. This step is executed only once by Algorithm~\ref{algo:enum-distinct-embeddings}, to have thereafter a usable database for our main algorithm. After this step we are left with the $13$ embeddings ($11$ for $K_5$ and $2$ for $K_{3,3}$) shown in Figure~\ref{fig:emb-data}. We will denote this set of embeddings $\Omega$. We can note that for $K_5$, there are 5 more embeddings on this surface than on the torus, described by Myrvold~\cite{MYRVOLD_TORUS_OBSTRUCTIONS}. We stress that these are the only possible ways, up to translation of the nodes, to draw these graphs on the Klein bottle, meaning that when extracting our subgraph, the given rotation system of this subgraph will necessarily correspond to exactly one of these drawings, again up to switch-equivalence.

\newcommand{\drawNodesCompleteInKlein}[5]{%
    \node[gnode](#1) at (2,2){};
    \node[gnode](#2) at (1,3){};
    \node[gnode](#3) at (3,3){};
    \node[gnode](#4) at (3,1){};
    \node[gnode](#5) at (1,1){};
}

\input{figures/unlabelled_Klein_embeddings.tex}

We setup the drawings of these embeddings to be convex (see Figure~\ref{fig:emb-data-convex}), in order to use a Tutte-like algorithm to
place the remaining nodes. Note that having strictly convex embeddings as starting base is not mandatory especially as the extracted subgraph homeomorphic to $K_5$ or $K_{3,3}$ could potentially have nodes of degree 2 which would lead to non-strictly convex embeddings.
The following Theorem ensures that the final drawing is equivalent to the rotation system given in input, since it proves that for each of the locally planar and orientable faces there exists a unique non-intersecting drawing.

\input{figures/unlabelled_Klein_embeddings_convex.tex}

\begin{theorem} \label{th:tutte-unique}
Let $G=\structure{V,E}$ be a graph, and $f$ a convex 2-cell region of the Klein bottle bounded by a fixed cycle $C$ of G. Let $H$ be the subset of the nodes of $G$ that has to be embed in $f$. If $G$ is 3-connected, then there exist a unique planar embedding of the subgraph of $G$ induced by $H \cup C$ where all the nodes of $H$ are in the same side of $C$.
\end{theorem}

\begin{proof}
First, let $v$ be a node on $C$ and $u$ a node adjacent to $v$. The node $v$ can potentially appear multiple times on $C$, however each of the possible attachments corresponds to a distinct position in $\pi_v$. Knowing $\pi_v$ allows to know on which occurrence of $v$ in $C$ $u$ must be attached. Thus, the ambiguity on the multiplicity of any node on the boundary of $f$ can be cleared up. 
Let assume now the $v$ can be embedded in two different faces, the one dictated by $\rot$ and another face of $\Pi(G-u)$ which contradicts $\rot$. 
Since $G$ is 3-connected, $u$ has at least three neighbours, let say $w_1$,$w_2$ and $w_3$ be those neighbours. If both embeddings are locally planar, both faces have on their boundary all the $w_i$. Let $P$ be a path in $G-u$ starting and ending on $C$ and containing all the $w_i$. Wlog, let $w_1$ and $w_3$ be such $w_2$ is included in the subpath of $P$ starting from $w_1$ and ending with $w_3$. Let assume removing the nodes $w_1$ and $w_3$ does not disconnect $G$ since it is 3-connected, then there must be a path connecting $w_2$ to $C$. However if $u$ can be embedded in both faces, the edges $(u,w_1)$ and $(u,w_3)$ must intersect this path, which is a contradiction.
\end{proof}

\begin{algorithm}[ht!]
\caption{Drawing from Rotation System} \label{alg:straight-line-klein}
\begin{algorithmic}[1]
\Require $G = \structure{V,E},\rot,\Omega$
\Ensure $\Gamma(G)$
\State $H \gets \texttt{KURATOWSKI\_SUBGRAPH}(G)$
\State $\tilde{H} \gets \texttt{SMOOTHED}(H)$
\For{$\Pi(K) \in \Omega$}\label{algo3:search-base-begin}
    \If{$\Pi(K) \iso \Pi(G)$}
        \State $\Gamma(\tilde{H}) \gets \Gamma(K)$
    \EndIf\label{algo3:search-base-end}
\EndFor
\For{$v \in V$} \label{alg:klein-draw:chain-nodes-start}
    \If{$v \in H$ \textbf{and} $v \notin \tilde{H}$}
        \State $u,w \gets \texttt{GET\_CHAIN\_ENDPOINTS}(v)$
        \State $\gamma_v \gets (\gamma_u+(\gamma_{u\rightarrow w}-\gamma_u)*\frac{P_{uw}.index(v)}{|P_{uw}|})\mod 1$
    \EndIf
\EndFor \label{alg:klein-draw:chain-nodes-end}
\For{$(u,v) \in E$} \label{alg:klein-draw:brides-start}
    \If{$u \in H$ \textbf{and} $v \in H$ \textbf{and} $(u,v) \notin H$}
        \State $\delta_{(u,v)} \gets  \floor{\gamma_{u \rightarrow v}}$ 
    \EndIf
\EndFor \label{alg:klein-draw:bridge-end}
\For{$v \in V$} \label{alg:klein-draw:faces-start}
    \If{$v \notin H$}
        \State $F \gets \texttt{GET\_FACE}(\Pi(G),H,v)$
        \State $\gamma_v \gets (\sum_{u \in F}{\gamma_u/|F|})\mod 1$ \label{alg:mid-face}
    \EndIf
\EndFor \label{alg:klein-draw:faces-end}
\State $\Gamma_0 \gets \Gamma(G)$
\State $\Gamma_1 \gets \texttt{TUTTE}(\Gamma_0)$
\State $i \gets 0$
\While{$\Gamma_i \neq \Gamma_{i+1}$}
    \State $\Gamma_{i+2} \gets \texttt{TUTTE}(\Gamma_{i+1})$ \label{alg:tutte}
    \State $i \gets i+1$
\EndWhile
\end{algorithmic}
\end{algorithm}

At this point we have all the main ingredients for our drawing algorithm. We assume to have the following routines:
\begin{itemize}
    \item \texttt{KURATOWSKI\_SUBGRAPH} which extract a subgraph of $G$ homeomorphic to a Kuratowski subgraph.
    \item \texttt{SMOOTHED} which takes a graph and returns a new one where all nodes of degree 2 have been replaced by an edge linking their two neighbours.
    \item \texttt{GET\_CHAIN\_ENDPOINTS} which takes a node of degree 2 and returns the two endpoints of the chain on which the node is.
    \item \texttt{GET\_FACE} computes, according to the order of the adjacency lists of the already fixed nodes, the face in which a node has to be embedded.
    \item \texttt{TUTTE} which moves all the non-fixed nodes to the barycenter of the positions of its neighbours. See~\cite{TUTTE_PLANAR_ALGO} for more details about Tutte's algorithm.
\end{itemize}

Moreover, the computation of the center of a face (line~\ref{alg:mid-face}) and the \texttt{TUTTE} (line~\ref{alg:tutte}) subroutine are performed according to the edges shifts as said before. For a face, we consider one point of the face in the square and get the coordinates of the other nodes by running through the edges and keeping track of the shifts of the edges. For the \texttt{TUTTE} routine, if a node $v$ would exit the $[0,1]$ square with this routine, its $\gamma_v$ and the shifts of its incident edges would have to be managed accordingly to maintain the coherence of $\Gamma(G)$.

\begin{proposition}
Algorithm~\ref{alg:straight-line-klein} runs in linear time in the number of edges of the input graph $G$.
\end{proposition}
\begin{proof}
The Kuratowski subgraph extraction can be done in linear time (see~\cite{WILLIAMSON_EMBEDING_GRAPH_ON_PLANE,WILLIAMSON_KURATOWSKI_SUBGRAPH} for details). The search for the base drawing (lines \ref{algo3:search-base-begin}-\ref{algo3:search-base-end}) is a comparison with a finite set of small graphs and hence, can be considered in constant time. Then, the drawing of the remaining nodes of the Kuratoswki subgraph (lines~\ref{alg:klein-draw:chain-nodes-start}-\ref{alg:klein-draw:chain-nodes-end}) can be done linearly if we don't recompute the chains each time. The drawing of the edges not present in the Kuratowsi subgraph but whose endpoints are (lines~\ref{alg:klein-draw:brides-start}-\ref{alg:klein-draw:bridge-end}) needs only a run through the boundary of the faces. Finding the faces of the remaining nodes (lines~\ref{alg:klein-draw:faces-start}-\ref{alg:klein-draw:faces-end}) can also be done linearly thanks to memoization. Finally, applying Tutte's algorithm on the non-fixed nodes is also linear, one can stop once each node has been moved at least once to ensure linearity. Thus, the total complexity of this algorithm is therefore in $O(n+m)$ with $n$ the number of nodes and $m$ the edges.
\end{proof}

\section{Conclusions}

We have highlighted the interest of having an algorithm for constructing straight-line representations of graphs on the Klein bottle. We proposed an algorithm to compare two general rotation systems and another one to enumerate all the possible ones of a given graph on a given surfaces, the latter using notions coming from the domain of signed graphs to be slightly more optimized than a naive one. We presented also and above all an algorithm to draw a graph, given its rotation system, and build a representation equivalent to it.

This work can be extended along several directions. The first one consists in characterizing more classes of graphs which are non-toroidal but embeddable in the Klein bottle. Another interesting theoretical question would be to have more tight bounds for the number of distinct unlabelled embeddings of a graph on the Klein bottle or on other surfaces, following on from works already done on similar questions~\cite{CHAPUY_COUNTING_UNICELLULAR_MAPS,GROSS_ENUMERATING_GRAPH_EMBEDDINGS}. Finally, one could study if our algorithm can be extended to surfaces of higher genera or if the flat representation of these surfaces, with more sided polygons, would lead to insoluble problems.

%% file: figures/unlabelled_Klein_embeddings.tex
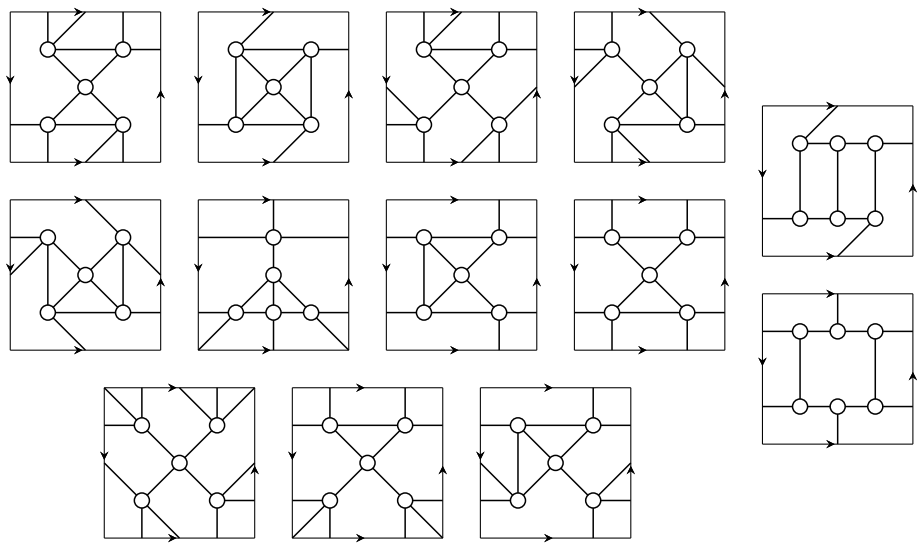
\begin{figure}
\centering
\begin{tikzpicture}[scale=1,decoration={markings,mark=at position 0.48 with {\arrow{stealth}}}]
\begin{scope}[xshift=0cm,yshift=0cm,scale=.5]
    \flatKleinBottle
    \drawNodesCompleteInKlein{0}{1}{3}{2}{4}
    \draw[gedge](0)--(1)--(3)--(0)--(4)--(2)--(0);
    \draw[gedge](1,4)--(1)--(2,4);
    \draw[gedge](3,4)--(3)--(4,3);
    \draw[gedge](0,1)--(4)--(1,0);
    \draw[gedge](2,0)--(2)--(3,0);
\end{scope}
\begin{scope}[xshift=2.5cm,yshift=0cm,scale=.5]
    \flatKleinBottle
    \drawNodesCompleteInKlein{1}{2}{3}{0}{4}
    \draw[gedge](2,4)--(2)--(4)--(0)--(3)--(2)--(1)--(4)--(0,1);
    \draw[gedge](2,0)--(0)--(1)--(3)--(4,3);
\end{scope}
\begin{scope}[xshift=5cm,yshift=0cm,scale=.5]
    \flatKleinBottle
    \drawNodesCompleteInKlein{0}{1}{2}{3}{4}
    \draw[gedge](0,2)--(4)--(0)--(2)--(1)--(0)--(3)--(4,2);
    \draw[gedge](1,4)--(1)--(2,4);
    \draw[gedge](3,4)--(2)--(4,3);
    \draw[gedge](0,1)--(4)--(1,0);
    \draw[gedge](2,0)--(3)--(3,0);
\end{scope}
\begin{scope}[xshift=7.5cm,yshift=0cm,scale=.5]
    \flatKleinBottle
    \drawNodesCompleteInKlein{1}{4}{3}{2}{0}
    \draw[gedge](0,2)--(4)--(1)--(2)--(4,1);
    \draw[gedge](2,0)--(0)--(1)--(3)--(2,4);
    \draw[gedge](0,3)--(4)--(1,4);
    \draw[gedge](1,0)--(0)--(2)--(3)--(4,2);
\end{scope}
\begin{scope}[xshift=0cm,yshift=-2.5cm,scale=.5]
    \flatKleinBottle
    \drawNodesCompleteInKlein{1}{4}{3}{2}{0}
    \draw[gedge](0,3)--(4)--(0)--(2)--(3)--(4,2);
    \draw[gedge](0,2)--(4)--(1)--(2)--(4,1);
    \draw[gedge](2,0)--(0)--(1)--(3)--(2,4);
\end{scope}
\begin{scope}[xshift=2.5cm,yshift=-2.5cm,scale=.5]
    \flatKleinBottle
    \node[gnode](0) at (2,1){};
    \node[gnode](1) at (2,2){};
    \node[gnode](2) at (3,1){};
    \node[gnode](3) at (2,3){};
    \node[gnode](4) at (1,1){};
    \draw[gedge](0,3)--(3)--(4,3);
    \draw[gedge](0,1)--(4)--(0)--(2)--(4,1);
    \draw[gedge](2,0)--(0)--(1)--(3)--(2,4);
    \draw[gedge](0,0)--(4)--(1)--(2)--(4,0);
\end{scope}
\begin{scope}[xshift=5cm,yshift=-2.5cm,scale=.5]
    \flatKleinBottle
    \drawNodesCompleteInKlein{0}{4}{1}{3}{2}
    \draw[gedge](0,3)--(4)--(0)--(3)--(3,0);
    \draw[gedge](0,1)--(2)--(0)--(1)--(3,4);
    \draw[gedge](4,3)--(1)--(4)--(2)--(3)--(4,1);
\end{scope}
\begin{scope}[xshift=7.5cm,yshift=-2.5cm,scale=.5]
    \flatKleinBottle
    \drawNodesCompleteInKlein{0}{4}{1}{3}{2}
    \draw[gedge](0,3)--(4)--(1)--(4,3);
    \draw[gedge](0,1)--(2)--(3)--(4,1);
    \draw[gedge](1,0)--(2)--(0)--(1)--(3,4);
    \draw[gedge](3,0)--(3)--(0)--(4)--(1,4);
\end{scope}
\begin{scope}[xshift=0cm+1.25cm,yshift=-5cm,scale=.5]
    \flatKleinBottle
    \drawNodesCompleteInKlein{0}{4}{1}{3}{2}
    \draw[gedge](0,4)--(4)--(0)--(3)--(4,2);
    \draw[gedge](1,0)--(2)--(0)--(1)--(3,4);
    \draw[gedge](2,4)--(1)--(4,4);
    \draw[gedge](0,2)--(2)--(2,0);
    \draw[gedge](3,0)--(3)--(4,1);
    \draw[gedge](0,3)--(4)--(1,4);
\end{scope}
\begin{scope}[xshift=2.5cm+1.25cm,yshift=-5cm,scale=.5]
    \flatKleinBottle
    \drawNodesCompleteInKlein{0}{4}{1}{3}{2}
    \draw[gedge](0,0)--(2)--(0)--(1)--(4)--(0)--(3)--(4,0);
    \draw[gedge](3,4)--(1)--(4,3);
    \draw[gedge](0,1)--(2)--(1,0);
    \draw[gedge](3,0)--(3)--(4,1);
    \draw[gedge](0,3)--(4)--(1,4);
\end{scope}
\begin{scope}[xshift=5cm+1.25cm,yshift=-5cm,scale=.5]
    \flatKleinBottle
    \drawNodesCompleteInKlein{0}{2}{1}{3}{4}
    \draw[gedge](0,3)--(2)--(0)--(1)--(2)--(4)--(0)--(3)--(4,2);
    \draw[gedge](3,4)--(1)--(4,3);
    \draw[gedge](3,0)--(3)--(4,1);
    \draw[gedge](0,1)--(4)--(0,2);
\end{scope}
\begin{scope}[xshift=10cm,yshift=-1.25cm,scale=.5]
    \flatKleinBottle
    \node[gnode](0) at (1,3){};
    \node[gnode](1) at (2,3){};
    \node[gnode](2) at (3,3){};
    \node[gnode](3) at (1,1){};
    \node[gnode](4) at (2,1){};
    \node[gnode](5) at (3,1){};
    \draw[gedge](2,4)--(0)--(1)--(2)--(5)--(2,0);
    \draw[gedge](0,1)--(3)--(4)--(5);
    \draw[gedge](4,3)--(2);
    \draw[gedge](0)--(3);
    \draw[gedge](1)--(4);
\end{scope}
\begin{scope}[xshift=10cm,yshift=-2.5cm-1.25cm,scale=.5]
    \flatKleinBottle
    \node[gnode](0) at (1,3){};
    \node[gnode](1) at (2,3){};
    \node[gnode](2) at (3,3){};
    \node[gnode](3) at (1,1){};
    \node[gnode](4) at (2,1){};
    \node[gnode](5) at (3,1){};
    \draw[gedge](0,3)--(0)--(1)--(2)--(4,3);
    \draw[gedge](0,1)--(3)--(4)--(5)--(4,1);
    \draw[gedge](2,4)--(1);
    \draw[gedge](2,0)--(4);
    \draw[gedge](0)--(3);
    \draw[gedge](2)--(5);
\end{scope}

\end{tikzpicture}
\caption{Unlabelled embeddings of $K_5$ and $K_{3,3}$ on the Klein bottle.} \label{fig:emb-data}
\end{figure}

%% file: figures/unlabelled_Klein_embeddings_convex.tex
\begin{figure}
\centering
\begin{tikzpicture}[scale=1,decoration={markings,mark=at position 0.48 with {\arrow{stealth}}}]
\begin{scope}[xshift=0cm,scale=.5]
    \flatKleinBottle
    \node[gnode](0) at (2,2){};
    \node[gnode](1) at (2,3){};
    \node[gnode](3) at (3,3){};
    \node[gnode](2) at (2,1){};
    \node[gnode](4) at (1,1){};
    \draw[gedge](0)--(1)--(3)--(0)--(4)--(2)--(0);
    \draw[gedge](1.5,4)--(1)--(2,4);
    \draw[gedge](2.5,4)--(3)--(4,3);
    \draw[gedge](0,1)--(4)--(1.5,0);
    \draw[gedge](2,0)--(2)--(2.5,0);
\end{scope}
\begin{scope}[xshift=2.5cm,scale=.5]
    \flatKleinBottle
    \node[gnode](1) at (2,2){};
    \node[gnode](2) at (2,3){};
    \node[gnode](3) at (3,2){};
    \node[gnode](0) at (2,1){};
    \node[gnode](4) at (1,2){};
    \draw[gedge](2,4)--(2)--(4)--(0)--(3)--(2)--(1)--(4)--(0,2);
    \draw[gedge](2,0)--(0)--(1)--(3)--(4,2);
\end{scope}
\begin{scope}[xshift=5cm,scale=.5]
    \flatKleinBottle
    \node[gnode](0) at (2,2){};
    \node[gnode](1) at (2,3){};
    \node[gnode](2) at (3,3){};
    \node[gnode](3) at (3,1){};
    \node[gnode](4) at (1,1){};
    \draw[gedge](0,2)--(4)--(0)--(2)--(1)--(0)--(3)--(4,2);
    \draw[gedge](1.5,4)--(1)--(2.5,4);
    \draw[gedge](3,4)--(2)--(4,3);
    \draw[gedge](0,1)--(4)--(1.5,0);
    \draw[gedge](2.5,0)--(3)--(3,0);
\end{scope}
\begin{scope}[xshift=7.5cm,scale=.5]
    \flatKleinBottle
    \node[gnode](1) at (2,2){};
    \node[gnode](4) at (1,3){};
    \node[gnode](3) at (3,3){};
    \node[gnode](2) at (3,2){};
    \node[gnode](0) at (2,1){};
    \draw[gedge](0,2)--(4)--(1)--(2)--(4,1.5);
    \draw[gedge](2.5,0)--(0)--(1)--(3)--(2.5,4);
    \draw[gedge](0,2.5)--(4)--(1.5,4);
    \draw[gedge](1.5,0)--(0)--(2)--(3)--(4,2);
\end{scope}
\begin{scope}[xshift=0cm,yshift=-2.5cm,scale=.5]
    \flatKleinBottle
    \node[gnode](1) at (2,2){};
    \node[gnode](3) at (2.5,3){};
    \node[gnode](2) at (3,2){};
    \node[gnode](0) at (2,1){};
    \node[gnode](4) at (1,1.5){};
    \draw[gedge](0,1.33)--(4)--(0)--(2)--(3)--(4,2.66);
    \draw[gedge](0,1.75)--(4)--(1)--(2)--(4,2.25);
    \draw[gedge](2.5,0)--(0)--(1)--(3)--(2.5,4);
\end{scope}
\begin{scope}[xshift=2.5cm,yshift=-2.5cm,scale=.5]
    \flatKleinBottle
    \node[gnode](0) at (2,1){};
    \node[gnode](1) at (2,2){};
    \node[gnode](2) at (3,1){};
    \node[gnode](3) at (2,3){};
    \node[gnode](4) at (1,1){};
    \draw[gedge](0,3)--(3)--(4,3);
    \draw[gedge](0,1)--(4)--(0)--(2)--(4,1);
    \draw[gedge](2,0)--(0)--(1)--(3)--(2,4);
    \draw[gedge](0,0)--(4)--(1)--(2)--(4,0);
\end{scope}
\begin{scope}[xshift=5cm,yshift=-2.5cm,scale=.5]
    \flatKleinBottle
    \node[gnode](0) at (2.5,2){};
    \node[gnode](4) at (1,2.5){};
    \node[gnode](1) at (3,3){};
    \node[gnode](3) at (3,1){};
    \node[gnode](2) at (1,1.5){};
    \draw[gedge](0,2.75)--(4)--(0)--(3)--(3,0);
    \draw[gedge](0,1.25)--(2)--(0)--(1)--(3,4);
    \draw[gedge](4,2.75)--(1)--(4)--(2)--(3)--(4,1.25);
\end{scope}
\begin{scope}[xshift=7.5cm,yshift=-2.5cm,scale=.5]
    \flatKleinBottle
    \drawNodesCompleteInKlein{0}{4}{1}{3}{2}
    \draw[gedge](0,3)--(4)--(1)--(4,3);
    \draw[gedge](0,1)--(2)--(3)--(4,1);
    \draw[gedge](1,0)--(2)--(0)--(1)--(3,4);
    \draw[gedge](3,0)--(3)--(0)--(4)--(1,4);
\end{scope}
\begin{scope}[xshift=0cm+1.25cm,yshift=-5cm,scale=.5]
    \flatKleinBottle
    \drawNodesCompleteInKlein{0}{4}{1}{3}{2}
    \draw[gedge](0,4)--(4)--(0)--(3)--(4,2);
    \draw[gedge](1,0)--(2)--(0)--(1)--(3,4);
    \draw[gedge](2,4)--(1)--(4,4);
    \draw[gedge](0,2)--(2)--(2,0);
    \draw[gedge](3,0)--(3)--(4,1);
    \draw[gedge](0,3)--(4)--(1,4);
\end{scope}
\begin{scope}[xshift=2.5cm+1.25cm,yshift=-5cm,scale=.5]
    \flatKleinBottle
    \drawNodesCompleteInKlein{0}{4}{1}{3}{2}
    \draw[gedge](0,0)--(2)--(0)--(1)--(4)--(0)--(3)--(4,0);
    \draw[gedge](3,4)--(1)--(4,3);
    \draw[gedge](0,1)--(2)--(1,0);
    \draw[gedge](3,0)--(3)--(4,1);
    \draw[gedge](0,3)--(4)--(1,4);
\end{scope}
\begin{scope}[xshift=5cm+1.25cm,yshift=-5cm,scale=.5]
    \flatKleinBottle
    \node[gnode](0) at (2.25,1.75){};
    \node[gnode](2) at (1.5,2.5){};
    \node[gnode](1) at (2,3){};
    \node[gnode](3) at (3,1){};
    \node[gnode](4) at (1,2){};
    \draw[gedge](3)--(0)--(2)--(4)--(0)--(1)--(2);
    \draw[gedge](2.5,4)--(1)--(4,2.33);
    \draw[gedge](0,2.5)--(4)--(0,1.66);
    \draw[gedge](4,1.2)--(3)--(2.5,0);
    \draw[gedge](2)--(0,2.8);
    \draw[gedge](3)--(4,1.5);
\end{scope}
\begin{scope}[xshift=10cm,yshift=-1.25cm,scale=.5]
    \flatKleinBottle
    \node[gnode](0) at (2,3){};
    \node[gnode](1) at (2.3,2.3){};
    \node[gnode](2) at (3,2){};
    \node[gnode](3) at (1,2){};
    \node[gnode](4) at (1.7,1.7){};
    \node[gnode](5) at (2,1){};
    \draw[gedge](2,4)--(0)--(1)--(2)--(5)--(2,0);
    \draw[gedge](0,2)--(3)--(4)--(5);
    \draw[gedge](4,2)--(2);
    \draw[gedge](0)--(3);
    \draw[gedge](1)--(4);
\end{scope}
\begin{scope}[xshift=10cm,yshift=-3.75cm,scale=.5]
    \flatKleinBottle
    \node[gnode](0) at (1,2.7){};
    \node[gnode](1) at (2,3.2){};
    \node[gnode](2) at (3,2.7){};
    \node[gnode](3) at (1,1.3){};
    \node[gnode](4) at (2,0.8){};
    \node[gnode](5) at (3,1.3){};
    \draw[gedge](0,2.7)--(0)--(1)--(2)--(4,2.7);
    \draw[gedge](0,1.3)--(3)--(4)--(5)--(4,1.3);
    \draw[gedge](2,4)--(1);
    \draw[gedge](2,0)--(4);
    \draw[gedge](0)--(3);
    \draw[gedge](2)--(5);
\end{scope}

\end{tikzpicture}
\caption{Strictly convex unlabelled embeddings of $K_5$ and $K_{3,3}$ on the Klein bottle.} \label{fig:emb-data-convex}
\end{figure}
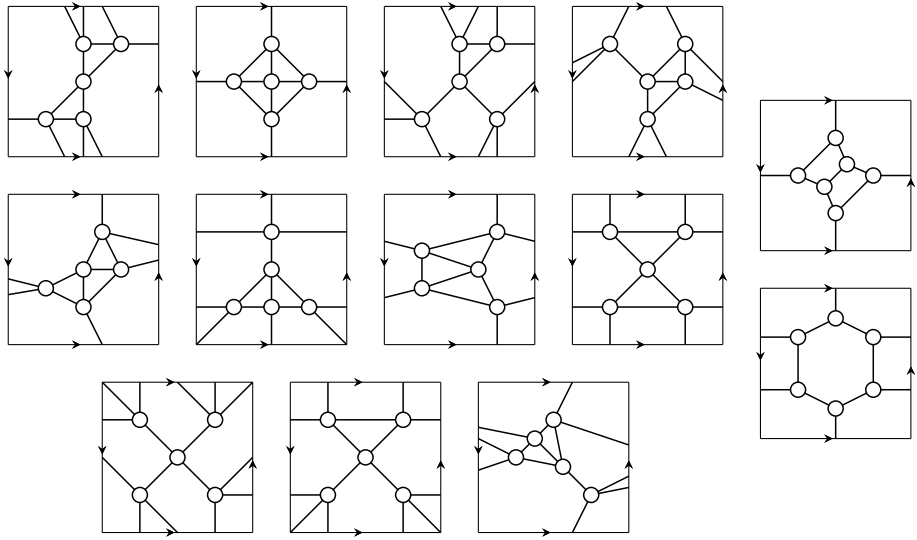